\documentclass{llncs}
\usepackage[utf8]{inputenc}
\usepackage[pdftex]{graphicx}
\usepackage[ruled,linesnumbered,vlined]{algorithm2e}
\usepackage{amsmath}
\newcommand{\argmin}{\operatornamewithlimits{argmin}}

\DeclareGraphicsExtensions{.pdf,.png,.jpg}
\graphicspath{{./images/}}
\begin{document}

%
%
%
\mainmatter              
\title{Parallel Query in the Suffix Tree}
\titlerunning{Parallel Query in the Suffix Tree}  
%
\author{Matevž Jekovec\inst{1} \and Andrej Brodnik\inst{1,2}}
\authorrunning{Matevž Jekovec \and Andrej Brodnik} 
%
%
\institute{$^1$University of Ljubljana,\\ Faculty of Computer and Information Science, SI\\
\email{\{matevz.jekovec,andrej.brodnik\}@fri.uni-lj.si},\\ WWW home page:
\texttt{http://lusy.fri.uni-lj.si/} \\ ~ \\
$^2$University of Primorska, \\Faculty of Mathematics, Natural Sciences and Information Technologies, SI}

\maketitle              

\begin{abstract}
Given a query string of length $m$, we explore a parallel pattern matching in a static suffix tree based data structure for $p \ll n$, where $p$ is the number of processors and $n$ is the length of the text. We present three results on CREW PRAM. The parallel query in the suffix trie requires $O(m + p)$ work, $O(m/p + \lg p)$ time and $O(n^2)$ space in the worst case. We extend the same technique to the suffix tree where we show it is inherently sequential in the worst case. Finallywe perform an interleaved parallel query which spends $O(m \lg p)$ work, $O(\frac{m}{p} \lg p)$ time and $O(n \lg p)$ space in the worst case.
\keywords{parallel pattern matching, suffix tree, PRAM}
\end{abstract}
\section{Introduction}

Pattern matching is one of the basic operations in text handling applications. We want to count occurrences of the given pattern in a text, locate them and retrieve the document content at specified position. Essentially there are two types of algorithms for pattern matching. The first type that constructs a finite automaton based on the pattern and process the text using this automaton \cite{Knuth1977,Boyer1977}. The least amount of work for finding all the occurrences of the pattern is bounded by the text size $n$. The second type first constructs a text index (e.g. the suffix array \cite{Manber1990}, the suffix tree \cite{McCreight1976}, or one of the compressed indexes, such as the FM index \cite{Navarro2007}). The queries are then answered by examining the index only (i.e. self-indexes), or a combination of the index and a subset of the original text. Required work of these algorithms is bounded by the pattern size $m$. In this paper we will focus on the suffix tree data structure, the fundamental text indexes and supports the pattern matching operations in $O(m)$ time.

Parallelism in text handling applications seems obvious since in practice we need to process a large amount of data. Vishkin \cite{Vishkin1985} designed the first optimal parallel algorithm of the first kind for locating occurrences of the pattern, which requires $O(n/p)$ time on CRCW PRAM for $p \leq n/\log m$. Later, work-optimal and work-time-optimal parallel string algorithms were introduced on EREW PRAM by Czumaj, Galil and others \cite{Czumaj1995}. Research on parallel algorithms of the second kind mostly focused on efficient construction of the text index. Farach-Colton et al.\ \cite{Farach-Colton2000} provided theoretical ground for optimal parallel suffix tree construction in the parallel disk-access model of computation. Latest practical algorithms for suffix tree construction with different worst case theoretical bounds include, for example, ERA \cite{Mansour2011,Jekovec2013} and Parallel Cartesian Tree \cite{Shun2014}. Researchers have also worked on both practically and theoretically sound algorithms for the suffix array and longest common prefix array construction \cite{Bingmann2013-2,Kaerkkaeinen2015}. In terms of the query speed research mostly focused on reducing cache misses when navigating the tree. Clark and Munro \cite{Clark1996} worked on succinct cache-efficient suffix trees. Ferragina and Grossi introduced String B-trees \cite{Ferragina1999} requiring optimal $O(\log_B n)$ cache misses in the worst case for any query, where $B$ is the block size in the external memory model. Brodal \cite{Brodal2006} designed a cache-oblivious variant of a data structure requiring the same optimal cache complexity. Notice the parallel trie navigation presented in this paper could also use a cache-oblivious organization of the trie in the backend. Demaine et al.\ \cite{Demaine2004} showed that arbitrary trie cannot be laid out to a memory and incur $\Omega(\log_B n)$ cache misses on queries in the worst case without some redundancy: Ferragina and Grossi propagated every $O(B)$-th value to the upper levels while Brodal used multiple layers of giraffe trees and bridges on top of the original trie.

In this paper we present a parallel query technique on a suffix trie and extend it to the suffix tree. Then, we explore a completely different approach to the query which employs interleaving the subqueries. Throughout the paper we will use CREW PRAM model with $p < 2m$ processors. Without loss of generality we assume $p=2^x$ for some integer $x$. The space complexity is expressed in words, if not stated otherwise. Our contributions are:
\begin{enumerate}
	\item A highly scalable parallel query algorithm in a suffix trie requiring $O(m + p)$ work, $O(m/p + \lg p)$ time and $O(n^2)$ space in the worst case.
	\item A proof that the parallel query approach used in a suffix trie achieves an inherently sequential execution time in the suffix tree.
	\item A parallel query algorithm in a layered interleaved suffix tree requiring $O(m \lg p)$ work, $O(\frac{m}{p} \lg p)$ time and $O(n \lg p)$ space in the worst case.
\end{enumerate}

To the best of our knowledge, this is the first result on the parallel queries in suffix tree based data structures employing $p \ll n$.

\section{Notation and Preliminaries}

We denote by $T$ the input text consisting of $n$ characters from a finite alphabet $\Sigma$. Further, we denote by $Q$ the query string of size $m$. We enumerate the elements of a list, an array and a string starting at $1$. By $p$ we denote the number of processors.

$X[i]$ denotes $i^{th}$ character of a string $X$. While $X[x_1, x_2]$ denotes a substring of $X$ ranging from the position $i_1$ to, including, the position $i_2$. If $i_1>i_2$, the resulting substring is empty. By $X[i, \cdot]$ we denote the substring $X[i,|X|]$, that is the suffix of $X$ starting at position $i$. If $X$ and $Y$ are strings, then $XY$ is their concatenation.

\subsection{Trie and Patricia trie}

\label{sec:trie-and-patricia-trie}

\emph{Trie} \cite{DeLaBriandais1959,Fredkin1960} is an ordered tree data structure used to store string-based keys. Every edge corresponds to one character. Each node in a trie represents a string of characters corresponding to the path from the root to this node. Trie can be used as a set, or as a dictionary, if we extend nodes to contain some value. The space complexity of a trie is $O(n \cdot l)$, where $n$ is the number of keys and $l$ the length of the longest key. By $\operatorname{child}(\omega, c)$ we denote a child of a node $\omega$ by taking an outgoing edge labeled by a character $c$, and by $\operatorname{parent}(\omega)$ we denote $\omega$'s parent. Further $\pi(\omega)$ denotes a sequence of nodes on a path from the root to the node $\omega$. Similarly $\pi(\omega_1, \omega_2)$ denotes a sequence of nodes on a path from node $\omega_1$ to node $\omega_2$ in a subtree of $\omega_1$.


\emph{Patricia tree} \cite{Morrison1968} is a path compressed trie, where a chain of nodes with a single child is merged into the final node of the chain. This node is either a leaf, or a node with more than one child. Each edge now corresponds to, instead of one character to a string of them. If we store two keys where the first key is a prefix of the second one, we need to add a unique \emph{delimiter character} or a unique sequence of characters at the end of the first key in order to discriminate between the two keys. Each node $\omega$ stores the first (discriminator) character and the length of the incoming edge's substring $\operatorname{skipvalue}(\omega) \leq n$. We define a \emph{cumulative skip value}, $|\omega|$, for node $\omega$ as the sum of all skip values from the root to $\omega$, including $\omega$ itself. The cumulative skip value is used during the query to determine which character of the query string needs to be compared to the character stored at the current edge in a patricia tree. The query is finished when we reach the first node with the cumulative skip value greater than the query length, or a leaf. Obviously, if any $\operatorname{skipvalue} > 1$, we must compare the query string to any of the suffixes stored in the resulting subtree to ensure the existence of the pattern in the text.

Assuming the keys are of constant size, patricia tree takes $O(n)$ space. Both trie and patricia tree have $\Theta(m)$ sequential query time.

\subsection{Suffix trie and Suffix tree}

\emph{Suffix trie} is a trie-based dictionary storing each suffix of the input text $T$, where each leaf $\nu$ stores $\operatorname{ref}(\nu)$ and represents a suffix $T[\operatorname{ref}(\nu), \cdot]$. The root node represents an empty string. We say a node in a suffix trie \emph{corresponds} to exactly one substring in the text occurring one or more times and consisting of characters on a path from a root to that node. 
Since suffix $X_i=T[i,\cdot]$ is $n-i$ characters long, the space complexity of the suffix trie is $\sum\limits_{i=1}^{n} |X_i| \in O(n^2)$.

\emph{Suffix tree} is a path compressed version of the suffix trie. Let $\omega$ denote a node in the suffix tree and let $\nu$ be the left-most leaf of the subtree rooted at $\omega$. We say $\omega$ \emph{corresponds} to substrings $A_i=T[ref(\nu), ref(\nu)+i]$ for all $i \in \lbrace | \operatorname{parent}(\omega)|+1 \ldots |\omega| \rbrace$, where the frequency of each substring in the text is equal. Obviously, if $\omega$ is a leaf, then $\omega$ corresponds to a suffix $A=T[\operatorname{ref}(\omega), \cdot]$. As usual, instead of the text itself, we can store constant size text references into the nodes of the suffix tree, so the space complexity of the data structure is $O(n)$.

In the suffix tree we define a \emph{suffix link} from node $\alpha'$ to node $\alpha$ iff their longest corresponding substrings are in relation $A=A'[2,\cdot]$. Since we inserted all the suffixes of the text, there exists an outgoing suffix link from each node in the suffix tree, except for the root node. Analogously, each node, except for the deepest leaf corresponding to $T$, contains an incoming suffix link.

\subsection{Perfect Hash Dictionary}

A \emph{perfect hash dictionary} is a hash table which requires $O(d)$ space for storing $d$ elements and uses $O(1)$ time in the worst case to answer, whether an element is in a dictionary. A static two-level perfect hash table has been introduced in \cite{Fredman1984}.

\section{Parallel query in a Suffix Trie}
\label{chap:parallel-query-in-sufix-trie}

We begin by introducing the fundamental suffix trie property used during the parallel query method.

\begin{lemma}
	\label{lemma:node_concatenation}
	Let $Q$ be a substring of $T$ and let the node $\omega$ in $T$'s suffix trie correspond to $Q$. Then, there exists a halving pair of nodes $(\alpha_1, \alpha_2)$ corresponding to the substrings $A_1=Q[1,\lceil|Q|/2\rceil]$ and $A_2=Q[\lceil|Q|/2\rceil+1, \cdot]$ respectively, where $Q=A_1A_2$.
\end{lemma}
\begin{proof}
	Since $Q$ is a substring of $T$, so are $A_1$ and $A_2$. Suffix trie contains all the substrings of the text, so $\alpha_1$ and $\alpha_2$ exist in the suffix trie (see Figure \ref{fig:parallel-suffix-trie}). \qed
\end{proof}

\begin{figure}[htb]
	\begin{center}
		\leavevmode
		\includegraphics[width=7cm]{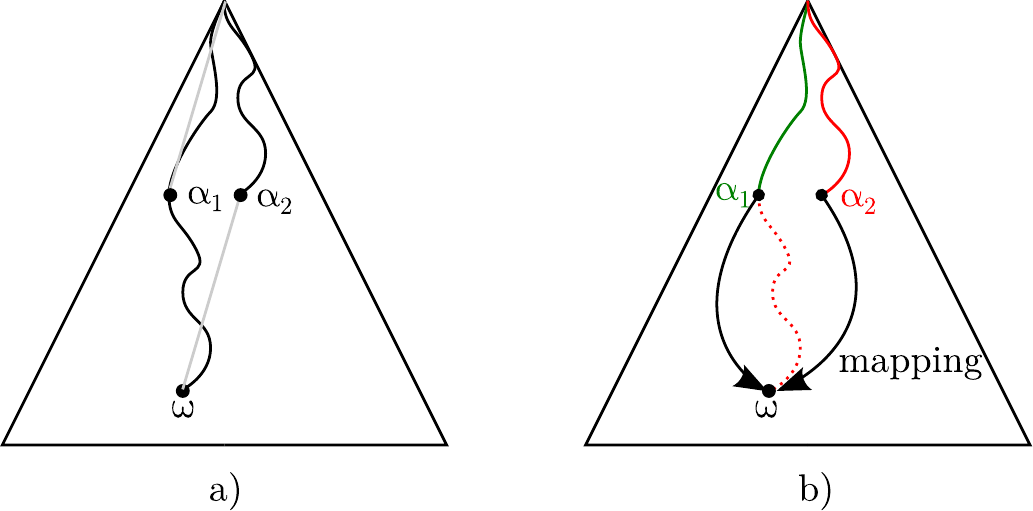}
	\end{center}
	\caption{The suffix trie. a) $\alpha_1$ and $\alpha_2$ are the halving pair of a substring which $\omega$ corresponds to. Gray lines mark the equivalence of the substrings obtained by $\pi(\alpha_2)$ and $\pi(\alpha_1,\omega)$. b) The mapping procedure of nodes $(\alpha_1, \alpha_2) \rightarrow \omega$. Notice $\pi(\alpha_1,\omega)$ is never navigated.}
	\label{fig:parallel-suffix-trie}
\end{figure}

The lemma provides a theoretical background for correct and efficient concatenation of two substrings of the text represented as nodes in the suffix trie. A consequence of Lemma \ref{lemma:node_concatenation} is that for each node $\omega$ in the suffix trie, there exists exactly one halving pair $\alpha_1, \alpha_2$ such that $|\alpha_1|=|\alpha_2|$ for even $|\omega|$, and $|\alpha_1|=|\alpha_2|+1$ for odd $|\omega|$. After the suffix trie is constructed, we create a dictionary, which maps a halving pair $(\alpha_1,\alpha_2) \rightarrow \omega$ for every $\omega$ in the suffix trie.
Notice the order of nodes $\alpha_1$ and $\alpha_2$ is important and reflects the order of the concatenated substrings.

\begin{algorithm}[htb]
	\SetKw{Continue}{continue}
	\KwIn{query string $Q$, suffix trie $\tau$, number of processors $p$}
	
	$Q_1 \ldots Q_p \leftarrow \operatorname{assign_p}(Q,p)$
	
	\For{$i \in \lbrace 1 \ldots p \rbrace$}{
		$\alpha_i \leftarrow \text{root node of }\tau$
	}
	
	\ForPar{$i \in \lbrace 1 \ldots p \rbrace$}{
		\While{$|\alpha_i| \leq |Q_i|$}{
			$\alpha_i \leftarrow \operatorname{child}(\alpha_i, Q_i[|\alpha_i|])$
		}
	}
	
	\For{$j \in [1, 2, 4 \ldots p/2]$}{
		\ForPar{$i=\lbrace 0 \ldots p/j-1 \rbrace$}{
			$\alpha_i \leftarrow \operatorname{probe}((\alpha_{ij+1}, \alpha_{ij+j+1}))$
		}
	}
	\Return{$\alpha_1$}
	\caption{Parallel query in the suffix trie employing $p$ processors.}
	\label{alg:parallel-suffix-trie-query}
\end{algorithm}

\begin{algorithm}[htb]
	\KwIn{query string $Q$, number of processors $p$}
	
	\ForPar{$i \in \lbrace 1 \ldots p \rbrace$}{
		$j_i \leftarrow \texttt{ReverseBits}(i-1)$ \tcp{MSB <-> LSB}
		$len_i \leftarrow |Q|/p + \lfloor j_i / (m \bmod p) \rfloor$ \tcp{$len_i$ equals $\lfloor |Q|/p \rfloor $ or $\lceil |Q|/p \rceil$}
	}
	
	$pos_i \leftarrow \texttt{ParPrefixSum}(len_i, p)$
	
	$Q_{1} \leftarrow Q[1, len_i]$
	
	\ForPar{$i \in \lbrace 2 \ldots p \rbrace$}{
		$Q_{i} \leftarrow Q[ pos_{i-1}, pos_{i-1}+len_i ]$
	}	
	
	\Return{$Q_1 \ldots Q_p$}
	\caption{Implementation of $\operatorname{assign_p}$.}
	\label{alg:parallel-suffix-trie-pAssign}
\end{algorithm}

The parallel query in the suffix trie is provided by Algorithm \ref{alg:parallel-suffix-trie-query}. To answer the query $Q$ we split it among $p$ processors (see Algorithm \ref{alg:parallel-suffix-trie-pAssign}). Each processor performs the query operation in the original, sequential suffix trie, with the assigned subquery. The processor independently navigates the tree. If any processor fails to take the corresponding branch during the search, the query does not appear in the text and we return an empty set of results. Otherwise, all processors successfully found nodes corresponding to their respective subqueries in time $O(m/p)$.



Let $\alpha_i$ be the resulting node of processor $i$ and $\omega$ the node corresponding to the query string $Q$. Our goal is to find a node $\omega$ from intermediate nodes $\alpha_i$ for $i=\{1 \ldots p\}$. We concatenate the substrings each $\alpha_i$ corresponds to pairwise: $\alpha_i^{(2)}=(\alpha_i, \alpha_{i+1})$ for every odd $i$. To obtain the whole path $\omega$, we continue the concatenation recursively $\alpha_i^{(3)}=(\alpha_i^{(2)}, \alpha_{i+2}^{(2)})$ for all $i \bmod 4=1$, in general $\alpha_i^{(j+1)} = (\alpha_i^{(j)}, \alpha_{i+2^j}^{(j)})$ for all $i \bmod {2^{(j)}} = 1$. In each step $j$, concatenations are done in parallel. After $\lg p$ steps, we reconstructed the whole path and obtained $\omega$.


\paragraph{Time and Work Complexity.}

The subquery lengths can be determined in constant parallel time. The parallel suffix trie navigation requires at most $O( m/p )$ parallel steps and $O(m)$ work in total. Final concatenation of $p$ subqueries requires $O(p)$ work and can be done in deterministic constant time per concatenation by using the perfect hash table. Using $p/2$ processors we can map $p$ subqueries to $p/2$ resulting nodes in a single step. Calculating the final $\omega$ node corresponding to the original query is done in $\lg p$ steps. The whole query requires $O(m + p)$ work and $O(m/p + \lg p)$ time. Concurrent reading is only required, if two processors traverse the same path in the suffix trie or access the same key in the hash table.

\paragraph{Space Complexity.}

When calculating the subquery lengths, we simultaneously need to keep $p$ words. Navigating the suffix trie requires constant space per processor. The dictionary size is one record for every node in the suffix trie. Overall we use $O(n^2)$ words as does the original suffix trie, which gives us:

\begin{theorem}
	Parallel query in suffix trie requires $O(m/p + \lg p)$ time on CREW PRAM, $O(n^2)$ space and $O(m + p)$ work.
\end{theorem}



\section{Parallel query in Suffix Tree}

The space reduction from $O(n^2)$ nodes in the suffix trie to $O(n)$ nodes in the suffix tree comes from removing the nodes with a single child (path compression). Consequently, Lemma \ref{lemma:node_concatenation} does not hold anymore since the node $\alpha_1$ of the halving pair might only have a single child in the suffix trie and it will not exist in the suffix tree.

In this section we correctly solve issues originating in the path compression and present a parallel query algorithm in the suffix tree. We begin by redefining the halving pair we store in the node mapping dictionary. Then we define the parallel query for $p=2$. Finally, we evaluate the work and time, and we find an example where the presented technique is inherently sequential in the worst case. Throughout this section we provide pseudocode to better describe our approach. To improve readability of the code however, we do not use any parallel constructs, but discuss the parallel implementation afterwards.

\subsection{Preprocessing}

We slightly loosen the equal work manner by extending the left subquery $A_1$ from Lemma \ref{lemma:node_concatenation} until we find a corresponding node $\alpha_1$ which exists in the suffix tree. Since the left subquery was extended, the right subquery $A_2$ needs to be shortened from the left accordingly. We will provide more details on subquery extension and shortening in the next section. First, we prove the existence of the new halving pair.

\begin{lemma}
	\label{lemma:node_concatenation_suffix_tree}
	For every node $\omega$ corresponding to a substring $Q$ in the suffix tree, there exists a halving pair $\alpha_1, \alpha_2$ corresponding to the substrings $A_1=Q[1,\hat{a}]$ and $A_2=Q[\hat{a}+1, \cdot]$ respectively for the smallest $\hat{a} \geq \lceil |Q|/2 \rceil$ in the suffix tree, where $Q=A_1A_2$.
\end{lemma}
\begin{proof}
	We start by assigning $\alpha_1=\omega$ and moving from $\omega$ towards the root until we reach the node corresponding to the smallest $\hat{a} \geq \lceil |Q|/2 \rceil$. In the worst case, we did not leave the initial node $\omega$ and $\alpha_1=\omega$. Otherwise, $\alpha_1$ is one of $\omega$'s ancestors.

	If node $\omega$ exists in the suffix tree, then $\omega$ is the leaf of the suffix tree, or have at least two children. Since the suffix tree contains all the suffixes, including the suffixes of $Q$, then there exists a node $\alpha_2$ corresponding to the substring $A_2=Q[\hat{a}+1, \cdot]$. $\alpha_2$ remains a leaf, if $\omega$ was a leaf, or has at least as many children as $\omega$ had. \qed
\end{proof}

We showed the halving pair $\alpha_1, \alpha_2$ exists in the suffix tree for every node $\omega$. Next, we must decide which substring of the resulting node $\omega$ in the suffix tree will the halving pair split. Recall the suffix tree definition: $\omega$ corresponds to, instead of one as in the suffix trie, to $\operatorname{skipvalue}(\omega) \in O(n)$ substrings. If we stored every corresponding halving pair of the potential queries in our node mapping dictionary, we would spend $O(n)$ space for a single node and $O(n^2)$ space for all the nodes eliminating the storage savings introduced by the suffix tree. Instead, we can only store a constant number of halving pairs per node in the suffix tree. To decide which pair to store, we begin with a lemma which captures an important consequence of the path compression in the suffix tree.

\begin{lemma}
	\label{lemma:suffix-tree-deeper-paths-sparser-than-shallower-ones}
	Let $\alpha, \alpha'$ be nodes in the suffix tree and $A, A'$ their longest corresponding substrings, such that $A'=CA$ for some prefix $C$. Also, let $\gamma$ be a node in the suffix tree, such that its longest corresponding substring is $C$. Then, $|\pi(\gamma,\alpha')| \leq |\pi(\alpha)|$.
\end{lemma}
\begin{proof}
	By counter example. Assume $|\pi(\gamma,\alpha')| > |\pi(\alpha)|$, then there exists a node $\beta$ such that $\beta \in \pi(\gamma,\alpha')$ and $\beta \notin \pi(\alpha)$. We denote by $B$ the longest substring corresponding to $\beta$. There exist two substrings $Bc_1$ and $Bc_2$ in the text for some characters $c_1,c_2$ which $\beta$ discriminates between. Since the suffix tree contains all the suffixes of the text, it also contains a suffix of $Bc_1$ and a suffix of $Bc_2$, and the nodes corresponding to these suffixes exist in $\pi(\alpha)$. This contradicts the initial assumption, since such $\beta$ does not exist. \qed
	
	
\end{proof}

\begin{proposition}
	\label{proposition:suffix_tree_default_splitting}
	For every node $\omega$ in the suffix tree and its corresponding substring $Q$, we store a pair $\alpha_1, \alpha_2$ corresponding to the substrings $A_1=Q[1,\hat{a}]$ and $A_2=Q[\hat{a}+1, \cdot]$ respectively for the smallest $\hat{a} \geq \lceil |Q|/2 \rceil$ in the suffix tree.
\end{proposition}

Suppose we store mapping $(\alpha_1\alpha_2) \rightarrow \omega$ as defined in Proposition \ref{proposition:suffix_tree_default_splitting} for every node $\omega$ in the suffix tree. In Lemma \ref{lemma:suffix-tree-deeper-paths-sparser-than-shallower-ones} we showed $|\pi(\alpha_2)| \geq |\pi(\alpha_1,\omega)|$. Follows, if a query is shorter that $|\omega|$, the right subquery might end, instead in $\alpha_2$, in one of the ancestors of $\alpha_2$, and we will not be able to find the mapping to $\omega$ in the dictionary. Notice however, if we navigated the tree sequentially, we would successfully find $\omega$, since the $|\pi(\alpha_1, \omega)|$ is smaller.

In order to correctly answer the query of any length corresponding to $\omega$, we store the halving pair of the \emph{shortest} substring corresponding to $\omega$, that is of length $|\omega|-\operatorname{skipvalue}(\omega)+1$.

\begin{proposition}
	\label{proposition:suffix_tree_short_splitting}
	For every node $\omega$ in the suffix tree and its corresponding substring $Q$, we store a halving pair $\beta_1, \beta_2$ corresponding to the substrings $B_1=Q[1,\hat{b}]$ and $B_2=Q[\hat{b}+1, |\omega|-\operatorname{skipvalue}(\omega)+1]$ respectively for the smallest $\hat{b} \geq \lceil (|\omega|-\operatorname{skipvalue}(\omega)+1)/2 \rceil$.
\end{proposition}

Denote $\beta_1$, $\beta_2$, and $\omega$ as in the Proposition \ref{proposition:suffix_tree_short_splitting}. In this case, a query pattern longer than $|\beta_1|+|\beta_2|$ requires special attention. Suppose the query is of length $|\omega|$. Observe $|\beta_1|+|\beta_2| > |\omega|/2$, because if the query were shorter we would have ended at one of the ancestors of $\omega$. Since we are not aware of which halving pair $\beta_1, \beta_2$ is stored in the dictionary, we gradually probe all feasible candidates during the navigation and remember the deepest one which existed in the dictionary. Figure \ref{fig:query-halving-node-propositions}a illustrates the query layout of Propositions \ref{proposition:suffix_tree_default_splitting} and \ref{proposition:suffix_tree_short_splitting}, and \ref{fig:query-halving-node-propositions}b the underlying suffix tree. In the rest of this section, we provide a parallel implementation of this gradual probing, and discuss the level of parallelism.

\begin{figure}[htb]
	\begin{center}
		\leavevmode
		\includegraphics[width=10cm]{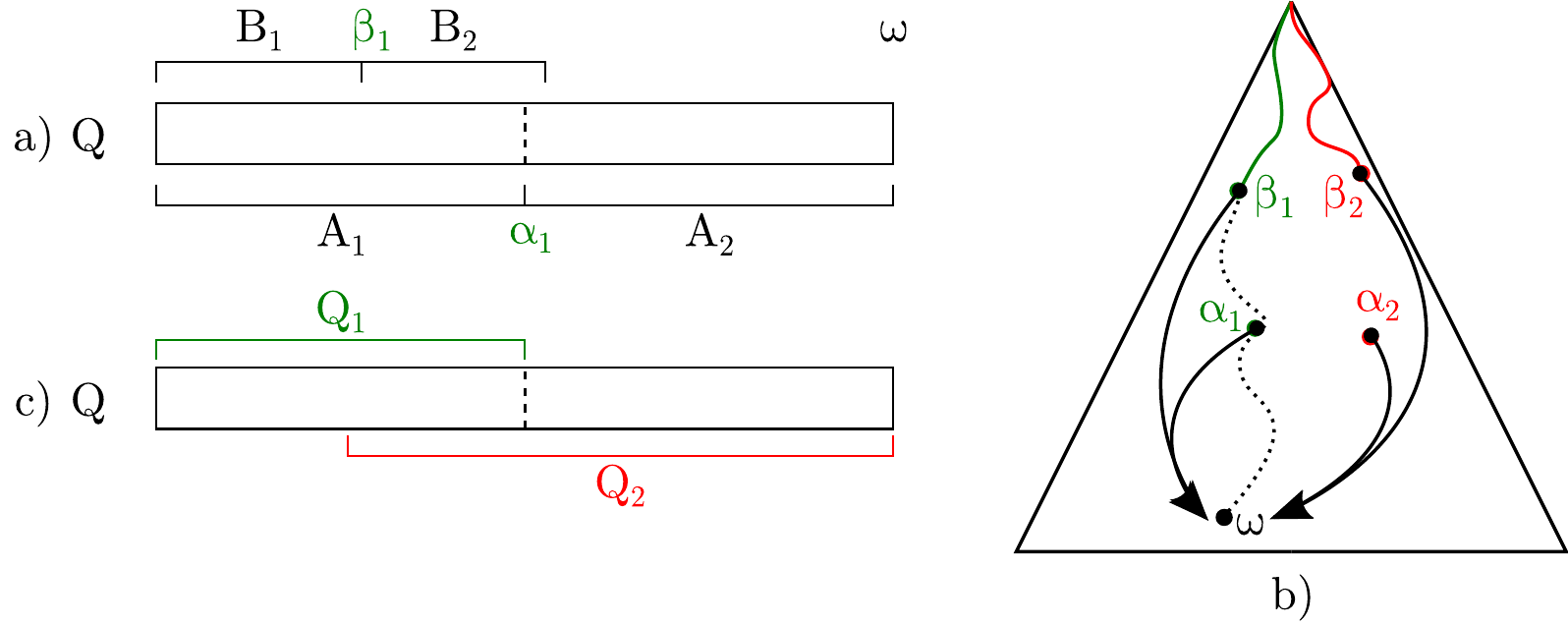}
	\end{center}
	\caption{a) Illustrated Propositions \ref{proposition:suffix_tree_default_splitting} and \ref{proposition:suffix_tree_short_splitting}. The vertical dots in the string illustrate the Position of $|\omega|/2$. b) Illustrated Propositions on the suffix tree, c) Subquery assignment in the adapted query method.}
	\label{fig:query-halving-node-propositions}
\end{figure}

\subsection{Query}
\label{sec:parallel-query-in-suffix-tree-query}

The query operation needs to find the deepest halving pair $\beta_1, \beta_2$ and map $(\beta_1, \beta_2) \rightarrow \omega$, such that the query string will correspond to one of $skipvalue(\omega)$ substrings. Acknowledging Proposition \ref{proposition:suffix_tree_short_splitting} and storing the halving pair of the shortest substring corresponding to $\omega$, the halving pairs of the following prefixes of $Q$ are potentially stored in the dictionary: $$Q[1,\lfloor|Q|/2\rfloor+1], Q[1, \lfloor|Q|/2\rfloor+2], \ldots ,Q{.}$$

We do not consider prefixes shorter than $\lfloor|Q|/2\rfloor+1$ characters since they correspond to ancestors of $\omega$ due to Lemma \ref{lemma:node_concatenation_suffix_tree}. Notice the halving pairs of multiple prefixes can be stored in the dictionary. In this case we need to consider the pair with the largest sum of cumulative skip values of the nodes in the pair. 

Performing a separate query for each prefix defined above would require $\Omega(|Q|^2)$ work overall. We use a more work efficient approach by taking the advantage of the suffix tree. We show how to perform a constant number of passes over $Q$ and check the relevant prefixes.

First, we assign the query string to the left and the right processor as follows. The shortest prefix of the query as defined above is $Q[1,\lfloor|Q|/2\rfloor+1]$. By acknowledging Lemma \ref{lemma:node_concatenation_suffix_tree}, its corresponding halving pair potentially stored in the dictionary corresponds to the subqueries $B_1=Q[1, \lceil|Q|/4\rceil+\hat{x}]$ and $B_2=Q[\lceil|Q|/4\rceil+\hat{x}+1, \lfloor|Q|/2\rfloor+1]$. $\hat{x}$ depends on the shape of the suffix tree and we cannot obtain it in advance. In the worst case $\hat{x}=0$, so we assign the right subquery $Q_2=Q[\lceil|Q|/4\rceil+1,\cdot]$ while keeping the original left subquery $Q_1=Q[1, \lceil|Q|/2\rceil]$ as we did in the parallel suffix trie. Notice $Q_2$ overlaps $Q_1$ for $\lceil|Q_1|/2\rceil$ characters, however we never concatenate $Q_1$ and $Q_2$ the way they were assigned. Instead we shorten the right subquery accordingly during the query procedure so the resulting concatenation of the subqueries is always a valid prefix of $Q$. In order to process the longest prefix, $Q$, the right ends of $Q_1$ and $Q_2$ remain $\lceil|Q|/2\rceil$ and $|Q|$ respectively. Figure \ref{fig:query-halving-node-propositions}c illustrates the subquery assignment. 

During the initial parallel navigation, we need to skip the prefixes shorter than $Q[1,\lfloor|Q|/2\rfloor+1]$ characters. Each processor navigates independently from each other the suffix tree according to the assigned subqueries and stops when each of them navigated $\lceil|Q|/4\rceil+1$ characters. If the cumulative skip value of each processor's node is not exactly $\lceil|Q|/4\rceil+1$, the processor stops at the last node which is strictly smaller than $\lceil|Q|/4\rceil+1$.

Follows the core of the query algorithm which repeats the following two steps. The \emph{navigation} step which we will describe later takes a child corresponding to the next character of either the left and the right subquery or only the right subquery, and the \emph{probe} step which effectively concatenates the substrings corresponding to the current nodes and checks whether the node corresponding to the concatenation exists in the dictionary. We will denote by $\beta_1$ and $\beta_2$ the current nodes reached by the left and the right processor respectively. In the first iteration we probe a halving pair $\beta_1, \beta_2$ in the dictionary, which corresponds to the shortest candidate prefix $Q[1,\lfloor|Q|/2\rfloor+1]$. We check whether such a mapping exists in the dictionary and, if it does, we remember the resulting node $\omega$. We continue with the next iteration of the navigation step and probe the dictionary for the new $\beta_1, \beta_2$. The algorithm continues until the cumulative sum $|\beta_1|+|\beta_2| \geq |Q|$. If any of the navigation steps failed due to non-existing edge in the suffix tree for the next character, the pattern does not exist in the text and we return the empty result set. Also, if a leaf $\nu$ was reached during the navigation step, at most a single occurrence of $Q$ exists in the text at position $ref(\nu)$, if the left subquery reached the leaf, or at position $ref(\nu)-|\beta_1|$, if the right subquery reached the leaf. Algorithm \ref{alg:parallel-suffix-tree-query} formally defines our approach.

\begin{algorithm}[htb]
	\SetKw{Continue}{continue}
	\KwIn{query string $Q$, suffix tree $\tau$}
	
	$\beta_{1} \leftarrow \text{root node of }\tau$
	
	$\beta_{2} \leftarrow \text{root node of }\tau$
	
	$Q_1 = Q[1, \lceil|Q|/2\rceil]$
	
	$Q_2 = Q[\lceil|Q|/4\rceil+1, \cdot]$
	
	\While{$|\beta_1|+|\beta_2| < |Q|$}{
		$(\beta_1, \beta_2) \leftarrow \operatorname{NavigateOne}(\beta_1, \beta_2, Q_1, Q_2, \tau)$
		
		\If{$|\beta_1|+|\beta_2| < \lceil|Q|/2\rceil$}{
			\Continue
		}
		
		$\omega' \leftarrow \operatorname{probe}((\beta_1, \beta_2))$
		
		\If{$\omega' \neq \emptyset$}{
			$\omega \leftarrow \omega'$
		}
	}
	\Return{$\omega$}
	\caption{Parallel query in the suffix tree for $p=2$.}
	\label{alg:parallel-suffix-tree-query}
\end{algorithm}

When we finish the query procedure and obtain the resulting node $\omega$, we need to check the skipped characters during the navigation due to the path compression used in the suffix tree. This step is already present in the original suffix tree data structure and is omitted in Algorithm \ref{alg:parallel-suffix-tree-query}. To check the skipped characters in parallel we access one of the leafs $\nu$ rooted in $\omega$'s subtree. Then, we use a parallel scan where the first processor compares $T[ref(\nu),ref(\nu)+\lceil|Q|/2\rceil]$ to $Q[1,\lceil|Q|/2\rceil]$, and the second processor compares $T[ref(\nu)+\lceil|Q|/2\rceil+1, ref(\nu)+|Q|]$ to $Q[\lceil|Q|/2\rceil+1, \cdot]$. If both subqueries match the text, the query substring truly exists in the text and we report all the leaves rooted in $\omega$'s subtree.

Next, we describe the navigation step in more detail. Assume $\beta_1$ and $\beta_2$ are the current nodes reached by the first and the second processor respectively according to the assigned subqueries, and $B_1$ and $B_2$ are the longest substrings corresponding to $\beta_1$ and $\beta_2$. The invariant of the navigation step is the following: $B_1B_2 = Q[1, |B_1|+|B_2|]$. We denote by $d_1$ and $d_2$ the skip value of the next $\beta_1$'s and $\beta_2$'s child respectively. If $|\beta_1| \geq |\beta_2| + d_2$, the navigation step will follow the edge to $\beta_2$'s child, extend the right subquery $B_2' = B_2 Q[|B_1|+|B_2|+1,|B_1|+|B_2|+d_2]$ while keeping $B_1$ intact and finish. Analogously, if $|\beta_1| < |\beta_2| + d_2$, the navigation step will follow the edge to $\beta_1$'s child and extend the left subquery $B_1' = B_1 Q[|B_1|+1, |B_1|+d_1]$. In this case however, $B_1'$ now overlaps $B_2$ for $d_1$ characters, or if $d_1 > |B_2|$, $B_1'$ covers the whole subquery $B_2$ and more. In the latter case, $B_2$ becomes an empty string corresponding to the root of the suffix tree and we are done. In the former case, in order to produce a correct concatenation of the subqueries, $B_2$ needs to be shortened from the left for $d_1$ characters. Shortening a substring corresponding to a node in the suffix tree can be done by following the suffix link from that node. By recursively following the suffix links we can shorten the substring for arbitrary number of characters. In the next paragraph we show how to perform this operation in $O(1)$ time. After the right subquery was shortened, $B_2' = B_2[d_1+1, |B_2|]$. Notice probing for the halving pair corresponding to $B_1'$ and $B_2'$ is not feasible, since $|B_1'| + |B_2'| = |B_1| + |B_2|$ and the previous $(\beta_1, \beta_2)$ was already probed in the last iteration. Instead, we continue with the next iteration of the navigate step. Algorithm \ref{alg:parallel-suffix-tree-navigate} formally defines the described navigation step and Figure \ref{fig:parallel-suffix-tree-query} illustrates the whole procedure.

\begin{algorithm}[htb]
	\KwIn{current nodes $\beta_1$, $\beta_2$, subqueries $Q_1$, $Q_2$}
	
	$d_2 \leftarrow \operatorname{skipvalue}(\beta_2)$
	
	\While{$|\beta_1| < |\beta_2|+d_2$}{
		$d_1 \leftarrow \operatorname{skipvalue}(\beta_1)$
		
		$\beta_1 \leftarrow \operatorname{child}(\beta_1, Q_1[|\beta_1|])$
		
		\If{$d_1 > |\beta_2|$}{
			$\beta_2 \leftarrow \text{root node of }\tau$
			
			\Return{$(\beta_1, \beta_2)$}
		}
		
		$\operatorname{shorten}(\beta_2, d_1)$
		
		$d_2 \leftarrow \operatorname{skipvalue}(\beta_2)$
	}
	
	$\beta_2 \leftarrow \operatorname{child}(\beta_2, Q_2[|\beta_2|])$
	
	\Return{$(\beta_1, \beta_2)$}
	\caption{The $\operatorname{navigateOne}$ function.}
	\label{alg:parallel-suffix-tree-navigate}
\end{algorithm}

\begin{figure}[htb] 
	\begin{center}
		\leavevmode
		\includegraphics[width=4cm]{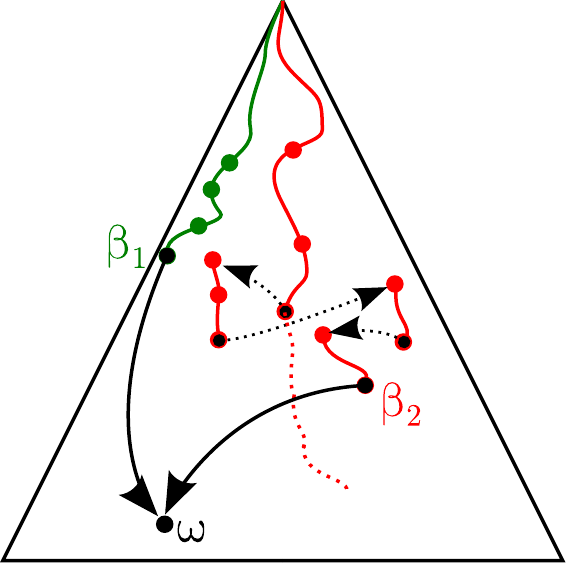}
	\end{center}
	\caption{Illustration of the parallel query for $p=2$ in the suffix tree. The solid lines represent visited paths in the suffix tree by the left and the right processor. Dotted arrows represent the followed suffix links. The dictionary contains a record $(\beta_1\beta_2) \rightarrow \omega$. Notice the dotted line below the path visited by the right processor shows a similarity to the concatenated chunks of paths obtained when following the suffix links.}
	\label{fig:parallel-suffix-tree-query}
\end{figure}

To preform $d$ steps over the suffix links in constant time we use the following method. We define the \emph{suffix links tree} such that for each node in the suffix tree, there exists one node in the suffix links tree. The edge between two nodes in the suffix links tree exists, if there exists a suffix link between the corresponding nodes in the suffix tree. Also, each node in the suffix tree contains a reference to the node in the suffix links tree, and the other way around. Observe that the node corresponding to the shortened substring for $d$ characters from the left is $d^{th}$ ancestor of the original node in the suffix links tree. We use the level ancestor technique to find the ancestor in constant time and maintain linear space (for details, consult to \cite{Bender2004}). The $\operatorname{shorten}$ function used in Algorithm \ref{alg:parallel-suffix-tree-navigate} employs this method.


\paragraph{Work Complexity.}

Assuming each character of a subquery corresponds to one unit of work, the amount of work to be invested by two processors is $\frac{1}{2}m + \frac{3}{4} m = \frac{5}{4} m$.

\paragraph{Time Complexity.}

To evaluate the time complexity of our algorithm, we need to reconsider dependencies of the second processor to the path constructed by the first processor. There exist two such dependencies: 1) the second processor requires information on the skip value of the nodes of the first processor in order to shorten the right subquery accordingly, and 2) if the condition $|\beta_1| < |\beta_2|+d_2$ in Algorithm \ref{alg:parallel-suffix-tree-navigate} is true, the second processor needs to wait the first processor to navigate the suffix tree and extend the left subquery, which leads to sequential execution in some cases.

We treat the first dependency as follows. The second processor requires the skip values of the nodes navigated by the first processor in order to correctly shorten the right subquery. However, it does not immediately require skip values of all the nodes on the first processor's path, but only of the first node that has larger cumulative sum than the cumulative sum of the currently navigated node by the second processor. By using a one node delay, we solve the dependency and we can still construct both paths in parallel in pipeline manner.

The second dependency requires inherently sequential execution in the worst case as a consequence of Lemma \ref{lemma:suffix-tree-deeper-paths-sparser-than-shallower-ones} and Proposition \ref{proposition:suffix_tree_short_splitting}. First, let processors independently of each other navigate the first $|Q|/2$ characters of the query. This costs $|Q|/4$ units of time. Then, let $\beta_2$ denote the current node of the second processor and assume it corresponds to the substring $Q[\lceil |Q|/4 \rceil+1, |Q|/2]$, where the skip value of the next child $d_2=\operatorname{skipvalue}(\operatorname{child}(\beta_2, Q[|Q|/2+1])) > \frac{1}{2} |Q|$. As a consequence of Proposition \ref{proposition:suffix_tree_short_splitting}, the second processor waits for the first processor to navigate the next node and extend the left subquery. The second processor shortens the right subquery accordingly, but the required skip value of a new node, now corresponding to $Q[\lceil |Q|/4 \rceil+1+d_1, |Q|/2]$, might still exceed the length of the left subquery. In the worst case, the first processor processes the whole left subquery, while the right processor only shortens the right subquery $|Q|/4$ times and spending $|Q|/4$ time. Finally, when the right subquery was shortened for the last time, the updated $d_2$ might suddenly become small and the second processor needs to navigate the rest of the assigned subquery spending $|Q|/2$ time in the worst case. Employing any number of processors, we will always spend $|Q|$ time overall in the worst case.

\begin{theorem}
Parallel query in the suffix tree using consecutive subqueries and linear space is an inherently sequential operation in the worst case.
\end{theorem}

\section{Interleaved suffix tree}

In the previous section we show that the parallel query method used in a suffix trie does not allow reasonable speed-ups in the suffix tree. Therefore, we explore a different approach to the parallel query where instead of splitting the query string into $p$ consecutive substrings, we split it into $p$-interleaving subsequences. To answer such an interleaved query, we need to construct a different, interleaved suffix tree and navigate this data structure instead of the original suffix tree. Finally, we map the obtained nodes from the interleaved suffix tree to a node in the original suffix tree and report the results.


\subsection{$k$-interleaved string}

\begin{definition}
	Given a string $X$ consisting of $n$ characters
	$X=c_1c_2 \ldots c_n$ we define $k$-interleaved subsequences $X_i$ of string $X$ for all $i \in \lbrace 1 \ldots k \rbrace$ such that
	$ X_i = c_i c_{i+k} c_{i+2k} \ldots c_{i+k \lfloor (n-i)/k \rfloor }{.} $
\end{definition}

For example the $2$-interleaved subsequences of the string $X=\texttt{ABRACADABRA}$ are $X_1=\texttt{ARCDBA}$ and $X_2=\texttt{BAAAR}$. In case of $k=1$, the $1$-interleaved subsequence $X_1$ is the original string $X$. In case of $k=n$, we obtain $n$-interleaved subsequences $X_i=c_i$ for $i\in \lbrace 1 \ldots n \rbrace$.



To deinterleave $k$-interleaved subsequences and construct the original sequence, we take one character at a time from subsequence $X_i$ for $i=[1 \ldots k]$ repeatedly until we reach the end of any subsequence. If the subsequences are arbitrarily long, the resulting string consists of the prefixes of the given subsequences.

\begin{definition}
	\label{def:interleaving-subsequences-arbitrary-long}
	Deinterleaving $k$ subsequences $X_i$ to resulting string $X$, where $|X_i|$ for $i \in \lbrace 1, \ldots, k \rbrace$ is arbitrary long, is done as follows:
	$$X[i]=X_{(i-1) \bmod k +1}[\lfloor (i-1)/k \rfloor+1]$$
	for all $i=\lbrace 1, \ldots, k \cdot min(|X_1|,\ldots,|X_k|) + \argmin\limits_j |X_j|\rbrace{.}$
\end{definition}
	
In the rest of the paper we will always deinterleave $2$-interleaved subsequences of arbitrary length. Acknowledging Definition \ref{def:interleaving-subsequences-arbitrary-long} and setting $k=2$, we observe the following.

\begin{property}
	\label{property:interleaving-of-non-equally-long-subsequences}
	Given $2$-interleaved subsequences $X_1$ and $X_2$ of arbitrary length, deinterleaving subsequences produces string $X$ of length $|X|=2 \cdot min(|X_1|, |X_2|)$, if $|X_1| \leq |X_2|$, or $|X|=2 \cdot min(|X_1|, |X_2|)+1$, if $|X_1| > |X_2|$. A compacted way of writing is $|X|=min(|X_1|,|X_2|)+min(|X_1|,|X_2|+1)$.
\end{property}


\subsection{$k$-interleaved suffix tree}

\begin{definition}
	$k$-interleaved suffix tree $\tau^{(k)}$ is a suffix tree containing all suffixes of all $k$-interleaved subsequences of the input text $T$.
\end{definition}

The construction of $k$-interleaved suffix tree is the same as of any other suffix tree, for example by inserting suffixes of each $k$-interleaved subsequence $T_i$ for $i=\lbrace 1 \ldots k\rbrace$ of the text $T$. Navigating a $k$-interleaved tree is also the same as navigating the original suffix tree. The leaves of the resulting node however, are the locations of the query string in the text where there is $k-1$ characters in the text between each character of the query string.

The delimiter character in $k$-interleaved suffix tree plays an additional role. Recall the delimiter character initially described in Section \ref{sec:trie-and-patricia-trie} used to discriminate between the strings where one string is a prefix of the other, or in case of the suffix tree, one suffix is a prefix of another suffix. The $k$-interleaved suffix tree, in addition to patricia tree, also needs to discriminate between equal suffixes which are part of different subsequences, and consequently appear at different locations in the text. To solve this, we append $k$ unique delimiter characters to the input text, so that each subsequence ends with a unique delimiter character (see Figure \ref{fig:interleave-suffix-trie}). We wrap up this subsection with two obvious properties:

\begin{property}
	Height of the $k$-interleaved suffix tree for the input text of length $n$ is at most $\lceil n/k \rceil$.
\end{property}

\begin{property}
	The number of leaves in $k$-interleaved suffix tree equals the number of leaves in the original suffix tree, excluding the leaves representing the delimiter characters.
\end{property}

\begin{figure}[htb]
	\begin{center}
		\leavevmode
		\includegraphics[width=8.5cm]{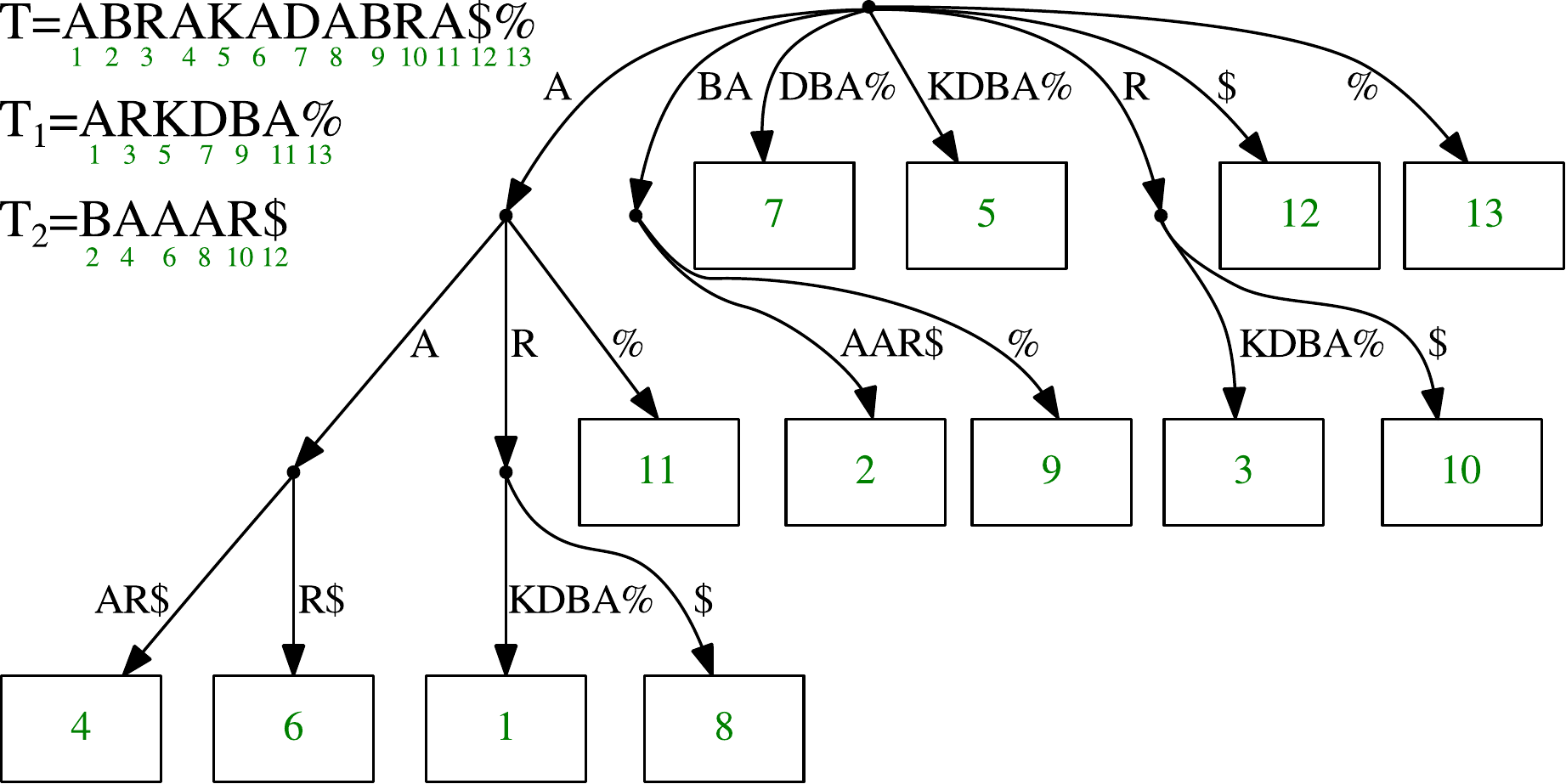}
	\end{center}
	\caption{A $2$-interleaved suffix tree for input text \texttt{ABRACADABRA}. Notice two delimiter characters \texttt{\$} and \texttt{\%} respectively.}
	\label{fig:interleave-suffix-trie}
\end{figure}




Notice, by keeping the same number of leaf nodes, the maximum height of the $k$-interleaved suffix tree in comparison to the original suffix tree is smaller because the width of the tree has been increased due to new delimiter characters in the alphabet. 

\subsection{Layered interleaved suffix tree}
\label{sec:layered-interleave-suffix-tree}

The layered interleaved suffix tree data structure consists of $\lg p$ layers, where at each layer $k$ for $k=\lbrace 2^0, 2^1, 2^2, \ldots, p \rbrace$ we store $k$-interleaved suffix tree $\tau^{(k)}$ for the same input text $T$, and a dictionary which maps a pair of nodes in layer $k$ to a node in layer $k/2$. The query using $j=2^x \leq p$ processors for some integer $x$ is done by navigating the layer $\tau^{(j)}$ and then mapping the obtained nodes from $\tau^{(j)} \rightarrow \tau^{(j/2)} \rightarrow \ldots \rightarrow \tau^{(1)}{.}$ The mapping merges partially navigated paths to obtain the resulting node. We will discuss the query procedure in detail in the next subsection. First, we formally describe how to map from one layer to the next one.

\begin{lemma}
	\label{lemma:interleave-suffix-tree-mapping}
	For each node $\omega \in \tau^{(k/2)}$ there exists a pair of nodes $\omega_1, \omega_2 \in \tau^{(k)}$
	such that deinterleaving the longest substrings which $\omega_1$ and $\omega_2$ correspond to, results in one of the substrings which $\omega$ corresponds to.
\end{lemma}

\begin{proof}
	Without loss of generality we set $k=2$. Assume one of the substrings which $\omega$ corresponds to is a substring of the input text $T[a,a+l]$ where $|\omega|-\operatorname{skipvalue}(\omega) \leq l \leq |\omega|$. Then, there exist two $2$-interleaved subsequences of $T$:
	$$X_1=T[a]T[a+2] \ldots T[a+2\lfloor l/2 \rfloor]$$
	$$X_2=T[a+1]T[a+3] \ldots T[a+2\lfloor (l-1)/2 \rfloor+1]{.}$$
	
	Notice $X_2$ is empty, if $\operatorname{skipvalue}(\omega)=1$.
	
	$X_1$ and $X_2$ are prefixes of two suffixes contained in $\tau^{(k)}$. Take the leaves corresponding to these suffixes and find the first nodes $\omega_1$ and $\omega_2$ respectively on a path from a root to a leaf which cumulative sum $|\omega_1| \geq |X_1|$ and $|\omega_2| \geq |X_2|$. Follows $\omega_1$ and $\omega_2$ exist in $\tau^{(k)}$ such that deinterleaving corresponding longest substrings is long $|\omega_1|+|\omega_2| \geq l$ characters and not shorter. This proves that $\omega$ or one of its descendants will correspond to the deinterleaved string.
	
	$\omega$ discriminates between at least two substrings in $T$ whose characters are different at position $|\omega|+1$. Consequently, $\omega_1$ or $\omega_2$ must also discriminate between at least two subsequences whose characters are different at position $\lfloor|\omega|/2\rfloor+1$. Following Definition \ref{def:interleaving-subsequences-arbitrary-long}, the length of the resulting string depends on the length of the shorter subsequence, formally \begin{equation*}
	\begin{split}
	|X| &= min(|X_1|,|X_2|)+min(|X_1|,|X_2|+1) \\
	&=min(|\omega_1|,|\omega_2|)+min(|\omega_1|,|\omega_2|+1) \leq |\omega|{.} \\
	\end{split}
	\end{equation*}
	This proves that $\omega$ or one of its predecessors will correspond to the resulting string. Since we excluded the predecessors in the previous paragraph, this concludes that exactly $\omega$ will correspond to deinterleaving of $\omega_1$ and $\omega_2$. 
	\qed 
\end{proof}

Each layer $k$ in layered interleaved suffix tree contains a dictionary which maps a pair of nodes $\omega_1,\omega_2$ in $\tau^{(k)}$ to a node $\omega$ in $\tau^{(k/2)}$ according to Lemma \ref{lemma:interleave-suffix-tree-mapping} for all nodes in $\tau^{(k/2)}$. Intuitively the mapping efficiently deinterleaves the longest subsequences corresponding to $\omega_1$ and $\omega_2$, to node $\omega$ corresponding to the resulting string. In turn the mapping fixates the $k-1$ characters in the text between each character of the query string, which occurred, if we queried the $k$-interleaved suffix tree as described in the previous subsection.

To construct the dictionary at each layer, we need to consider which pair $(\omega_1,\omega_2) \rightarrow \omega$ to store. Because of the path compression, node $\omega$ corresponds to $O(n)$ substrings. Storing each pair of nodes $\omega_1, \omega_2$ for which their corresponding longest substrings deinterleave to one of the substrings corresponding to $\omega$ would require $O(n^2)$ space per each layer. We use the same reasoning as in Propositions \ref{proposition:suffix_tree_default_splitting} and \ref{proposition:suffix_tree_short_splitting} in the previous section. We find and store a pair of nodes which corresponding substrings deinterleave to the \emph{shortest} substring $\omega$ corresponds to.

\begin{definition}
	\label{def:interleave-suffix-tree-dictionary}
	Dictionary of $\tau^{(k)}$ is a dictionary, which for each node $\omega \in \tau^{(k/2)}$ maps a pair of nodes $\omega_1, \omega_2 \in \tau^{(k)}$ to $\omega$ for the smallest $|\omega_1|$ and $|\omega_2|$ such that deinterleaving the longest substrings which $\omega_1$ and $\omega_2$ correspond to results in one of the substrings $\omega$ corresponds to.
\end{definition}

Storing such pairs, we assure nodes $\omega_1,\omega_2$ will always be reachable by the shortest possible query corresponding to $\omega$. This, however, imposes additional work during the query phase in order to correctly find mappings for the longer query strings which we describe next.


\subsection{Parallel Query in Layered Interleave Suffix Tree}

We assign $p$-interleaved subsequences of $Q$ to processors $1 \ldots p$ respectively. Then, each processor, independently navigates $\tau^{(p)}$ according to the assigned subsequence. 

Let $\pi_i^{(j)}$ denote the navigated path by processor $i$ in $\tau^{(j)}$. Initially we navigated paths $\pi_i^{(p)}$. We need to interleave the longest substrings which the nodes on these paths correspond to and obtain the paths at layer $p/2$. At each layer $j$, the interleaving is done pairwise as follows. We interleave each path $\pi_i^{(j)}$ with $\pi_{i+1}^{(j)}$ for $i=\lbrace 1, 3, \ldots j-1\rbrace$ and obtain $j/2$ new paths at layer $j/2$. We continue interleaving the paths recursively until we reach the first layer, that is the original suffix tree.

To efficiently interleave all substrings which the nodes of two paths $\pi_i^{(j)}$ and $\pi_{i+1}^{(j)}$ correspond to, we probe the dictionary at layer $j$ described in the previous subsection. The probed key consists of one node from $\pi_i^{(j)}$ and the second node from $\pi_{i+1}^{(j)}$. If the probe is successful, we append the resulting node to the new path $\pi_{\lceil i/2 \rceil}^{(j/2)}$ in the increasing cumulative skip value order.

Next, we explore the required number of probes in the dictionary in order to construct a valid path at the next layer. The trivial upper bound probing all possible pairs of nodes in both paths is $\left( m/j \right)^2$, where $m$ denotes the original query length. We reduce the number of probes to $O(m/j)$ in the following way.

\begin{lemma}
	\label{lemma:interleaving-paths-time}
	Deinterleaving paths $\pi^{(j)}_i$ and $\pi^{(j)}_{i+1}$ requires $|\pi^{(j)}_i|+|\pi^{(j)}_{i+1}|$ probes in the dictionary.
\end{lemma}

\begin{proof}
	We start by $\omega_1 \leftarrow \pi_i^{(j)}[1]$ and $\omega_2 \leftarrow \pi_{i+1}^{(j)}[1]$. Then, we compare the cumulative skip values of children of $\omega_1$ and $\omega_2$, and we follow an edge to a child with the smaller one due to Property \ref{property:interleaving-of-non-equally-long-subsequences}. If the values of both children are equal, we follow an edge to $\omega_1$'s child. We repeat the procedure until we reached the end of both paths. Each time we follow an edge, we perform a probe $\omega_1,\omega_2$ to the dictionary. Since each time we follow one node, the procedure requires exactly $|\pi_i^{(j)}|+|\pi_{i+1}^{(j)}|$ probes to the dictionary.
	\qed
\end{proof}

Deinterleaving procedure introduced in the proof of Lemma \ref{lemma:interleaving-paths-time} is embarrassingly parallel. Each of $j$ processors is assigned a consecutive $m/j$ nodes of paths. Notice the paths need to be aligned according to the cumulative skip value of the nodes. This way, we optimally employ $j$ processors and construct the path at the next layer. Figure \ref{fig:interleave-st-parallel-query} illustrates the navigation in $2$-interleaved suffix tree and the required probes in the dictionary.

\begin{figure}[htb]
	\begin{center}
		\leavevmode
		\includegraphics[width=8cm]{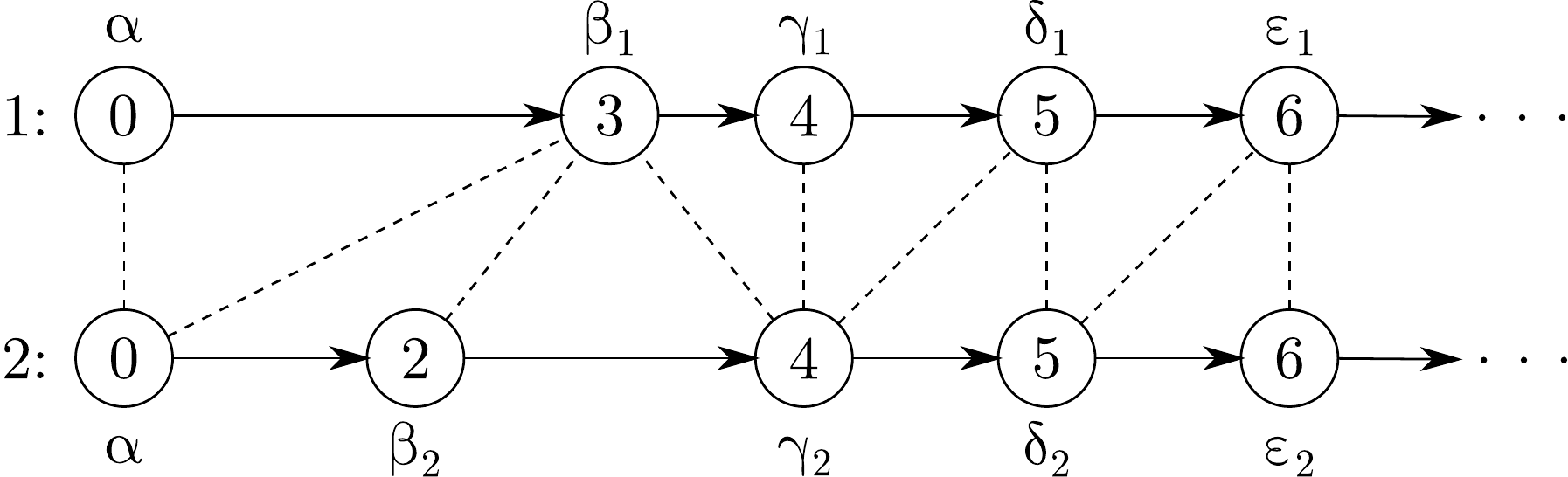}
	\end{center}
	\caption{Dictionary probes during the parallel query in $2$-interleaved suffix tree for $p=2$. The two paths on the figure are visited nodes by processors $1$ and $2$. $\alpha$ is the root node of $2$-interleaved suffix tree common to both processors, while the other nodes are separated. Numbers inside nodes denote the cumulative skip value of the node. Dotted lines illustrate the probes to the dictionary during the parallel query.}
	\label{fig:interleave-st-parallel-query}
\end{figure}

Finally, when we constructed a path at the first layer $\pi_1^{(1)}$, the deepest node on this path is the resulting node of the query. We check whether skipped characters during the query match the query string and report all the leaves of the subtree rooted at the resulting node.

%

\paragraph{Time and Work Complexity.}
The assignment of interleaved subqueries to $p$ processors can be done implicitly. Each processor navigates the $p$-interleaved suffix tree independently which requires $O(m/p)$ time and $O(m)$ work. Finally, deinterleaving intermediate paths requires $O(\frac{m}{p} \lg p)$ time on CREW PRAM and $O(m \lg p)$ work.

In detail, we deinterleave each pair of paths at the top layer in parallel as described before. Also, each pair of paths deinterleaves independently from the others. This way we employ all $p$ processors and spend $O(m/p)$ time. Notice each probe to dictionary takes constant time by using the perfect hash table. At the next layer, we deinterleave $p/2$ paths of length $\leq 2m/p$. In general, at each layer we use $p$ processors and spend $O(m/p)$ time. After $\lg p$ steps, we construct the final path in the original suffix tree.

\paragraph{Space Complexity.}
Each $k$-interleaved suffix tree requires $O(n)$ space. The layered interleaved suffix tree consists of $\lg p$ layers of $k$-interleaved suffix trees and dictionaries requiring $O(n \lg p)$ space overall. Notice, if we support parallel query for the fixed $p$ and not smaller, we can only keep $p$-interleaved suffix tree at layer $p$, $\lg p-1$ dictionaries, and the original suffix tree. Asymptotically this still requires $O(n \lg p)$ space. The intermediate paths during the query procedure require $O(m+p)$ temporary space.

\begin{theorem}
	Parallel query in layered interleaved suffix tree requires $O(\frac{m}{p} \lg p)$ time on CREW PRAM and $O(n \lg p)$ space.
\end{theorem}

\section{Conclusion}

We have explored a parallel query in suffix tree based data structures where the number of processors $p \ll n$. We have presented two algorithms. The first one uses as a data structure a suffix trie. It splits the query string into $p$ consecutive subqueries, navigates the underlying suffix trie, and merges the intermediate nodes into the final node. This requires $O(m + p)$ work, $O(m/p + \lg p)$ time and $O(n^2)$ space in the worst case. The second algorithm extends the same principle to suffix trees where we correctly solve issues concerning the path compression. However, the subquery overlapping required to solve the path compression issues introduced dependencies between the navigated paths. In the worst case, no parallelism can be employed and an inherently sequential execution is performed. Finally we have presented a layered interleaved suffix tree data structure. Instead of processing the query in a consecutive subquery manner, we use interleaving. This approach requires $O(m \lg p)$ work, $O(\frac{m}{p} \lg p)$ time and $O(n \lg p)$ space in the worst case. To the best of our knowledge, the presented algorithms are the first parallel algorithms for pattern matching requiring the amount of work not related to $n$. The number of processors in previous solutions assumed $p$ is of order of $n$.

There exists an open question whether we reached the space-work lower bound for the pattern matching problem. We also haven't provided any parallel cache complexity analysis of our query algorithms, perhaps using the cache-oblivious string dictionary \cite{Brodal2006} or the string B-tree \cite{Ferragina1999} as a starting point. From the applied point of view, an interesting research question would be whether the presented algorithms improve cache performance in comparison to traditional sequential query, since each processor accesses $O(m/p + \lg p)$ nodes instead of $O(m)$ during the suffix tree navigation and there is a better chance the accessed nodes might have already been in the cache due to temporal locality.

\bibliographystyle{splncs03}
\bibliography{jekovec}

\clearpage


\end{document}